\documentclass{article}
\usepackage[margin=1in]{geometry}
\usepackage{amsmath}
\usepackage{amssymb}
\usepackage{amsthm}
\usepackage{bm}
\usepackage{graphicx}
\usepackage[usenames,dvipsnames]{xcolor}
\usepackage{cite}
\usepackage{subfig}
\usepackage{hyperref}
\usepackage{balance}
\usepackage[capitalise]{cleveref}
\usepackage{threeparttable}
\usepackage{multirow}
\usepackage{booktabs}

\hypersetup{pdfencoding=unicode}

\DeclareMathOperator*{\argmin}{arg\,min}

\newtheorem{proposition}{Proposition}

\title{Design and Analysis of a Synthetic Prediction Market using Dynamic Convex Sets}

\author{
Nishanth Nakshatri\footnote{Dept. of Computer Science and Engineering, Pennsylvania State University, University Park, PA 16802} \and 
Arjun Menon\footnotemark[1] \and 
C. Lee Giles\footnote{College of Information Sciences and Technology, Pennsylvania State University, University Park, PA 16802} \and 
Sarah Rajtmajer\footnotemark[2] \and 
Christopher Griffin\footnote{Applied Research Laboratory, Pennsylvania State University, University Park, PA 16802}
}

\begin{document}
\maketitle


\begin{abstract}
We present a synthetic prediction market whose agent purchase logic is defined using a sigmoid transformation of a convex semi-algebraic set defined in feature space. Asset prices are determined by a logarithmic scoring market rule. Time varying asset prices affect the structure of the semi-algebraic sets leading to time-varying agent purchase rules. We show that under certain assumptions on the underlying geometry, the resulting synthetic prediction market can be used to arbitrarily closely approximate a binary function defined on a set of input data. We also provide sufficient conditions for market convergence and show that under certain instances markets can exhibit limit cycles in asset spot price. We provide an evolutionary algorithm for training agent parameters to allow a market to model the distribution of a given data set and illustrate the market approximation using two open source data sets. Results are compared to standard machine learning methods.
\end{abstract}




\section{Introduction}
Prediction markets in their current form trace their roots to the original studies by Hanson \cite{hanson1990market,hanson1991more,hanson1995could,ray1997idea} and since then have been studied and used extensively \cite{wolfers2004prediction,servan2004prediction,manski2006interpreting,berg2003prediction,wolfers2006prediction,dai2020wisdom,chakraborty2016trading}. For a survey of work in this area through 2007 see \cite{tziralis2007prediction}. In these markets, assets corresponding to future events (e.g., elections \cite{berg1997makes}, sports outcomes \cite{thaler1988anomalies} etc.) can be bought and sold thereby manipulating underlying asset prices. These asset prices can be interpreted as probabilities \cite{manski2006interpreting,wolfers2006interpreting} thereby providing a mechanism for event forecasting. Recent applications of prediction markets include forecasting infectious disease activity \cite{polgreen2007use}, evaluating scientific hypotheses \cite{almenberg2009experiment}, predicting the reproducibility of scientific work \cite{dreber2015using}, and aggregation of employee wisdom in a corporate setting \cite{cowgill2009using,gillen2012information}. 

In practice, many of these markets have been remarkably successful in efficiently aggregating information about uncertain future events \cite{smith2003constructivist}. There are a number of compelling explanations for this. Financial stakes incentivize participants to search for better information \cite{arrow2008promise} and the forecasts of more confident agents are weighted more heavily, where confidence is measured as willingness to risk more money \cite{goel2010prediction}. The efficient markets hypothesis suggests that the market price reflects available information at least as well as any competing method \cite{expectations1961theory}, although some have suggested that this hypothesis is not upheld in prediction markets \cite{manski2006interpreting}. Work has explored specific concerns about liquidity, price manipulation, outcome manipulation, bias, and their respective impacts on market efficiency \cite{tetlock2008liquidity,hanson2004manipulators,tetlock2007optimal,wolfers2006interpreting,sunstein2006infotopia,ottaviani2007outcome}. A separate thread of this research has studied the accuracy of prediction markets based on real versus play money, to disentangle the specific role of financial incentives (see, e.g., \cite{pennock2001real,servan2004prediction,rosenbloom2006statistical,gruca2008incentive}). The arrival of blockchain technologies has facilitated the development of decentralized prediction markets (e.g., \cite{Augur,Gnosis,Stox}), which benefit from the trust and transparency inherent in these ownerless peer-to-peer systems. Blockchain-based prediction markets offer anonymity for their traders \cite{clark2014decentralizing,heilman2016blindly}, support broad participation, and reduce single points of failure \cite{bentov2017decentralized}. Design of decentralized prediction markets is an ongoing area of research \cite{peterson2015augur,subramanian2017decentralized,wang2018preliminary}.  

Over the last decade, a body of work has emerged on so-called artificial (equivalently, synthetic) prediction markets. These are numerically simulated markets populated by artificial participants (agents) for the purpose of supervised learning of probability estimators \cite{barbu2012introduction}. Like their human-populated counterparts, artificial prediction markets have found a number of applications, including lymph node detection from CT scans \cite{barbu2013artificial} and early stage detection of epidemics from crowd-sourced data \cite{jahedpari2014artificial}. The theoretical promise of artificial markets was first explored by Chen and colleagues \cite{chen2008complexity,chen2010new,abernethy2011optimization}. They highlight the deep mathematical connections between prediction markets and learning, demonstrating that any cost function based prediction market with bounded loss can be interpreted as a no-regret learning algorithm \cite{chen2010new}. And, that every convex cost function based prediction market can be interpreted as a Follow the Regularized Leader algorithm with a convex regularizer \cite{abernethy2011optimization}.

In an initial construction put forward by Barbu and Lay \cite{barbu2012introduction} patterned after the Iowa Electronic Markets \cite{wolfers2004prediction}, each agent is represented as a budget and a simple betting function. During training, each agent's budget is updated based on the accuracy of its prediction for each training data point. The contract price for an outcome is an estimator of its class-conditional probability. These markets, authors found, were able to outperform random forest and implicit online learning in benchmark classification tasks. In follow-up work \cite{lay2012artificial}, the same authors generalized the market framework to support regression and reported similar gains in performance. Storkey and colleagues \cite{storkey2011machine,storkey2012isoelastic} develop an artificial prediction market with a different market mechanism, the so-called machine learning market. In their formulation, each agent purchases contracts for possible outcomes in order to maximize its own utility function. The equilibrium price of the contracts is computed by an optimization procedure. The market is shown to outperform standard classifiers on a number of machine learning benchmarks. A 2014 extension of this work \cite{hu2014multi} models agents using static risk measures. The authors demonstrate that the resulting market approaches a global objective, formally asserting the potential of the market to solve problems in machine learning. More recently, authors have proposed continuous artificial prediction markets \cite{jahedpari2017online} for online regression. These markets consider agents with adaptive trading strategies, using reinforcement learning to dynamically identify actions that maximize their own reward.

In this paper we study synthetic prediction markets in which the agents' purchase logic is governed by time-varying semi-algebraic sets. For the purposes of this work, we focus on convex semi-algebraic sets defined by ellipsoids in $\mathbb{R}^n$. Time variation of the set volume is governed by asset prices in the market. Agents specialize in the purchase of a single asset class and will only purchase an asset at time $t$ if an input feature vector is contained in the (time-varying) set defining the agent. We show the following:
\begin{enumerate}
    \item Given an arbitrarily large but finite labeled data set, we show how to construct a market that will perfectly assign to each input the appropriate output. This allows us to derive a form of universal approximation for our market structure.
    \item We provide a sufficient condition in terms of the underlying geometric structures for a market to converge to a single final price for all assets. 
    \item We show that the market can exhibit limit cycles and these limit cycles correspond to input data that lie near decision boundaries of agents. 
    \item We develop an evolutionary algorithm for training agent behavior in a market to represent a set of input data.
    \item We illustrate this algorithm using three open source data sets.
\end{enumerate}
Our results are complementary to the existing synthetic prediction market literature and establish a geometric foundation for building more complex prediction markets.

The remainder of this paper is organized as follows: In \cref{sec:Model} we discuss the synthetic prediction market model and establish relevant notation. Theoretical results on the prediction market are established in \cref{sec:Properties}. We discuss an algorithm for training a market to classify samples from a specific data set in \cref{sec:Training}. In \cref{sec:Results} we show empirical results on three open source machine learning data sets. Conclusions and future directions of research are presented in \cref{sec:Conclusions}.

\section{Binary Market Model}\label{sec:Model}



Let $\mathbb{Z}_+$ be the positive integers. Assume we have a binary option market with the two options denoted as Assets $0$ and $1$.  Assume $\mathbf{q}_t = (q^0_t,q^1_t) \in \mathbb{Z}_+^2$ units of (Asset 0, Asset 1) at time $t$ have been sold. A (binary option) market \cite{wolfers2006interpreting} $M$ consists of a set of agents $\mathcal{A} = \{a_1,\dots,a_n\}$ who buy (and sell) Assets $0$ and $1$ using policies $\{\gamma_1,\dots,\gamma_n\}$. If agent purchase policy $\gamma_i$ is conditioned on exogenous information $\mathbf{x} \in D \subseteq \mathbb{R}^n$ then, $\gamma_i : (\mathbf{q}_t,\mathbf{x}) \mapsto (r^0,r^1)$ and Agent $i$ purchases $r^0$ units of Asset $0$ and $r^1$ units of Asset $1$, thus causing a state update. 
When the market is conditioned on $\mathbf{x} \in D$ we denote it $M_\mathbf{x}$.

Assuming time passes discretely (is epochal) and we have an input $\mathbf{x} \in D$,  market $M_\mathbf{x}$ is a dynamical system $(\mathbb{Z}_+^2, \mathbb{Z}_+,\Gamma_\mathbf{x})$ where the dynamic $\Gamma_\mathbf{x}:\mathbb{Z}_+^2 \to \mathbb{Z}_+^2$ arises from the interaction of the individual policies $\{\gamma_1,\dots,\gamma_n\}$ and the conditional information $\mathbf{x}$. At any time $t$, the state $\mathbf{q}_t$ can be mapped into a pair of asset prices $\mathbf{p}_t = (p^0_t,p^1_t)$ that may be used in the policies of the agents in place of $\mathbf{q}_t$. 

\subsection{Market Details}
For the remainder of this paper, we will assume that $\Gamma_\mathbf{x}$ is fixed when given $\mathbf{x}$ and that an initial state $\mathbf{q}_0$ is given. We use the Logarithmic Market Scoring Rule (LMSR) \cite{hanson2007logarithmic} to aggregate estimates from a set of agents $\mathcal{A} = \{a_1,\dots,a_n\}$ and determine asset prices. Given state $(q^0_t,q^1_t)$, the current asset prices are computed using LMSR: 
\begin{gather*}
p^0_t = \frac{\exp{(\beta q^0_t)}}{\exp{(\beta q^0_t)} + \exp{(\beta q^1_t)}}\\
p^1_t = \frac{\exp{(\beta q^1_t)}}{\exp{(\beta q^0_t)} + \exp{(\beta q^1_t)}}.
\end{gather*}
This is the softmax function (Boltzmann distribution with constant $\beta = k/T$ for fixed $k$ and $T$) of the inputs $(q^0_t,q^1_t)$. The $\beta$ term is a liquidity factor \cite{lekwijit2018optimizing} that adjusts the amount the price will increase or decrease given a change in the asset quantities. By using a Boltzmann distribution, the prices can be interpreted as probabilities. 

The true asset purchase prices (trade costs) are not given by $\mathbf{p}_t$, since LMSR incorporates a market maker cost. The trade costs are given by:
\begin{gather*}
\kappa_t^0(\Delta q^0) = \frac{1}{\beta} \log\left(\frac{\exp[\beta (q_t^0+\Delta q^0)] + \exp[\beta q_t^1]}{\exp[\beta q_t^0] + \exp[\beta q_t^1]}\right)\\
\kappa_t^1(\Delta q^1) = \frac{1}{\beta} \log\left(\frac{\exp[\beta q_t^0] + \exp[\beta (q_t^1 + \Delta q^1)]}{\exp[\beta q_t^1] + \exp[\beta q_t^2]}\right),
\end{gather*}
where $\Delta q^i$ is the change in the quantify of Asset $i$ as a result of purchases defined by $\Gamma_\mathbf{x}$.

Let $P(\mathbf{x},t) = \mathbf{p}_t(\mathbf{x})$ assuming fixed $\mathbf{q}_0$ and $\Gamma_x$. The market converges to a price pair $\bar{\mathbf{p}}$ if:
\begin{equation}
\lim_{t \to \infty} P(\mathbf{x},t) = \bar{\mathbf{p}}.
\end{equation}
Convergence is not necessarily guaranteed in all markets, however for the markets we consider, we will show sufficient conditions for convergence to occur. 

Let $\phi:D \subseteq \mathbb{R}^n \to [0,1]$ be a binary function. Our objective is to construct $\Gamma_\mathbf{x}$, which defines a market $M$ and agents $\mathcal{A}$, so that:
\begin{equation}
\int_D |\phi(\mathbf{x}) - \bar{p}^1(\mathbf{x})|^2\,d\mathbf{x} < \epsilon,
\end{equation}
where $\bar{p}^1(\mathbf{x})$ is the long-run price of Asset 1 and $\epsilon > 0$ is a (small) error term. The left hand side yields the $L^2$ error when the price of Asset $1$ is used as an approximation function for $\phi$. We make this more precise in subsequent sections.

\subsection{Agent Purchase Policies}
Let $f(\mathbf{x};\bm{\theta})$ be a quasi-concave function parameterized by $\bm{\theta}$ with maximum at $\mathbf{0}$.  By this we mean a function that satisfies the inequality:
\begin{equation}
f(\lambda\mathbf{x}_1 + (1-\lambda)\mathbf{x}_2;\bm{\theta}) \geq \max\{f(\mathbf{x}_1;\bm{\theta}),f(\mathbf{x}_2;\bm{\theta})\}.
\end{equation}
If $\bm{\Theta}$ is a positive definite, diagonal matrix, then the quadratic function:
\begin{equation}
g(\mathbf{x};\bm{\theta}) = 1 - \mathbf{x}^T\bm{\Theta}\mathbf{x} = 1 - \sum_{j} \theta_i \mathbf{x}_i^2
\end{equation}
is such a function and the set:
\begin{equation}
\mathcal{E}_\mathbf{Q} = \{\mathbf{x} \in \mathbb{R}^n : g(\mathbf{x};\mathbf{Q}) \leq 0\}
\end{equation}
is an ellipsoid centered at $\mathbf{0}$ and oriented along the standard basis.

For the chosen quasi-concave function, define the translated function:
\begin{equation}
f(\mathbf{x};\mathbf{h},\bm{\theta}) = f(\mathbf{x} - \mathbf{h};\bm{\theta})
\end{equation}
In terms of the quadratic function this is just:
\begin{equation}
g(\mathbf{x};\mathbf{h},\bm{\Theta}) = (\mathbf{x} - \mathbf{h})^T\bm{\Theta}(\mathbf{x} - \mathbf{h}).
\end{equation}
Under these assumptions, $\bm{\Theta}$ defines a simple local metric that is used to determine how close the conditioning point $\mathbf{x}$ is to a reference point $\mathbf{h}$. 

Assume we are given a set of labeled training data  $\mathbf{H} = \{\mathbf{h}^i\}_{i=1}^N$ with labels $\mathbf{y} = \{y^i\}_{i=1}^N$ with $y^i \in \{0,1\}$. For each data point $\mathbf{h}_i$  in $\mathbf{H}$ (or possibly an appropriate subset of $\mathbf{H}$) with label $y^i$ define Agent $i$ who buys \textit{only} Asset $y^i$ ($i=0,1$). 
That is, we assume that Agent $i$ specializes in buying $y_i$. Given an input feature vector $\mathbf{x}$, Agent $i$ estimates the value of Asset $y_i$ using the  formula: 
\begin{equation}
\pi_t^i(\mathbf{x},p_t^{y^i};\mathbf{h}^i;\bm{\theta}^i,\alpha^i,w_p^i) =
\sigma[\alpha^i \cdot  f(\mathbf{x};\mathbf{h}^i,\bm{\theta}^i) + \theta_0^i + w_p^i (p_t^{y_i}-p^{y_i}_0)],
\label{eqn:Price}
\end{equation}
where $p_t^{y_i}$ is the price of Asset $y_i$ at time $t$, $\theta_0$ is a bias, $\alpha$ is a scaling factor and $\sigma$ is the logistic sigmoid function\footnote{A unit step function could be substituted with minimal change to the sequel.}. When using an ellipsoidal function, the exact formula is:
\begin{equation}
\pi^i_t(\mathbf{x},p^{y_i}_t;\mathbf{h}^i,\bm{\theta}^i,\alpha^i,w_p^i) = 
\sigma\left[
\alpha^i \cdot \left(1 - \sum_{j} \theta_j^i (x - h_j^i)^2 \right)+ w_p^i (p^{y_i}_t-p^{y_i}_0) + \theta_0^i\right].
\label{eqn:EllipsoidPrice}
\end{equation}
We note that if $\theta^i_j = 0$, then the ellipsoid structure is replaced (effectively) with a cylinder in $\mathbb{R}^n$. 

We assume Agent $i$ can only buy one unit of Asset $y_i$ at a time (per epoch). The agent logic defining $\gamma_i$ is then:
\begin{enumerate}
\item For $\Delta q^{y_i} = 1$, if 
\begin{displaymath}
\frac{1}{\kappa^{y_i}_t}\left(\pi^i_t\left(\mathbf{x},p_t^{y_i};\mathbf{h}^i;\bm{\theta}^i,\alpha^i,w_p^i\right) - \kappa^{y_i}_t[\Delta q^{y_i}]\right) \geq \tau,
\end{displaymath}
then the agent purchases a single unit of Asset $y_i$. Here $\tau \in [0,1)$ determines the opportunity cost considered by the agent. When $\tau = 0$, the agent purchases an asset precisely when it has sufficient funds and when it's estimated price is higher than the actual asset price.

\item Otherwise, the agent purchases nothing.
\end{enumerate}
For our model, each agent only buys when the conditioning data $\mathbf{x} \in D$ is close enough (in the derived metric) to its initialized data point $\mathbf{h}^i$. Thus, we are using the data set $\mathbf{H}$ to construct a covering of the set $D$ and then using that covering to construct the market and its dynamics.

\section{Properties of the Market}\label{sec:Properties}
In this section, we study the theoretical properties of markets in which agents have unlimited funds. 

\subsection{Approximation}
\begin{proposition} Let $\mathbf{H} = \left\{\mathbf{h}^i\right\}_{i=1}^N$ be a finite but arbitrarily large data set  with labels $\mathbf{y} = \{y^i\}_{i=1}^N$. Assume the data are separable; i.e., if $\mathbf{h}^i = \mathbf{h}^j$, then $y^i = y^j$. For all $\epsilon > 0$, there is a market $M$ with agents $\mathcal{A} = \{a^1,\dots,a^N\}$ such that for all $i = 1,\dots,N$:
\begin{equation}
    \lim_{t\to \infty} |p^1(\mathbf{h}^i,t) - y^i| < \epsilon,
    \label{eqn:Limit}
\end{equation}
where $p^1(\mathbf{x},t)$ is the price of Asset 1 in the market (the market spot price).
\label{prop:Approximation}
\end{proposition}
\begin{proof} Set $\tau = 0$. The fact that $\mathbf{H}$ is finite implies there is a set of open spheres centered at $\mathbf{h}^1,\dots,\mathbf{h}^N$ with radii $r_1,\dots,r_N$ so that:
\begin{equation}
    \mathbf{h}^j \in B_{r_i}(\mathbf{h}^i) \iff \mathbf{h}^j = \mathbf{h}^i
\end{equation}
From \cref{eqn:EllipsoidPrice}, 
for all $i$ and $j$, define $\theta^i_j = 1/r_i^2$. For all $i$ set $w_p^i = \theta_0^i = 0$. Assume that Agent $i$ purchases only Asset $y^i$. For Agent $i$ using \cref{eqn:EllipsoidPrice} the estimated price given $\mathbf{h}^i$ is constant and given by: 
\begin{equation}
    \pi^i(\mathbf{h}^i) = \pi_t^i(\mathbf{h}^i, p_t^{y^i};\mathbf{h}^i,\bm{\theta}^i,\alpha^i,w_p^i) = \sigma(\alpha^i) > \frac{1}{2}.
\end{equation}
Likewise, it is clear that for $\mathbf{h}^j \neq \mathbf{h}^i$:
\begin{equation}
    \pi^i(\mathbf{h}^j) = \pi_t^i(\mathbf{h}^j, p_t^{y_i};\mathbf{h}^i,\bm{\theta}^i,\alpha^i,w_p^i) < \frac{1}{2},
\end{equation}
since by construction:
\begin{displaymath}
1 - \sum_k \frac{\left(h^j_k - h^i_k\right)^2}{r_i^2} < 0.
\end{displaymath}
Set $\alpha^i = \alpha$ so that (by choice of $\alpha$) for all $i,j$:
\begin{align*}
    \pi^i(\mathbf{h}^i) &> 1 - \delta\\
    \pi^i(\mathbf{h}^j) & < \delta,
\end{align*}
for a $\delta \in (0,\epsilon/2)$. Such an $\alpha$ must exist because $\sigma$ is monotonic and bounded between $0$ and $1$. When $\mathbf{h}^i$ is used as the market input (i.e., $\mathbf{x} = \mathbf{h}^i$), then Agent $i$ will purchase one share of Asset $y^i$ per epoch until the first time $t^{(1)}$ when:
\begin{displaymath}
1 - \delta < \pi^i(\mathbf{h}^i) < \kappa^{y_i}_{t^{(1)}}.
\end{displaymath}
Choose $\beta$ small enough to ensure that at this point:
\begin{equation}
1 - \delta < p^{y_i}_{t^{(1)}} = \frac{e^{\beta t^{(1)}}}{1 + e^{\beta t^{(1)}}} < \kappa^{y_i}_{t^{(1)}}.
\label{eqn:BetaReq1}
\end{equation}
There are two possibilities.

\noindent{\textbf{Case I}:} For all $j$:
\begin{displaymath}
\pi^j(\mathbf{h}^i) < \delta < \kappa^{1-y_i}_{t^{(1)}}.
\end{displaymath}
In this case, the market converges to price $p^{y_i}_{t^{(1)}} > 1 - \delta > 1 - \epsilon$ as required.

\noindent{\textbf{Case II:}} There is at least one $j$ so that \begin{displaymath}
\pi^j(\mathbf{h}^i) > \kappa^{1-y_i}_{t^{(1)}}.
\end{displaymath}
At $t^{(1)}$ all such agents will purchase shares of asset $(1-y_i)$ and will continue to do so until $t^{(2)}$ at which point either Case I holds or Agent $i$ purchases again. In each case, assume $\beta$ is chosen small enough so that at time $t^{(2)}$:
\begin{equation}
p^{y^i}_{t^{(2)}} > 1 - 2\delta. 
\label{eqn:BetaReq2}
\end{equation}
This ensures that the purchases of the other agents cannot drive the price too far from $1 - \delta$. Such a $\beta$ must exist because asset price moves are monotonically decreasing in $\beta$. Since $\pi^i(\mathbf{h}^i)$ and $\pi^j(\mathbf{h}^i)$ are fixed for all time and $\mathbf{H}$ is finite, a smallest fixed value of $\beta$ must exist to make \cref{eqn:BetaReq1,eqn:BetaReq2} true for all time. (See \cref{fig:PriceDiff}.) We repeat the above logic to see that for time $t \geq t^{(1)}$, $p^{y_i}_t \in (1 - 2\delta, 1 - \delta)$ and \cref{eqn:Limit} holds. This completes the proof.
\end{proof}
\begin{figure}[htbp]
\centering
\includegraphics[width=0.45\columnwidth]{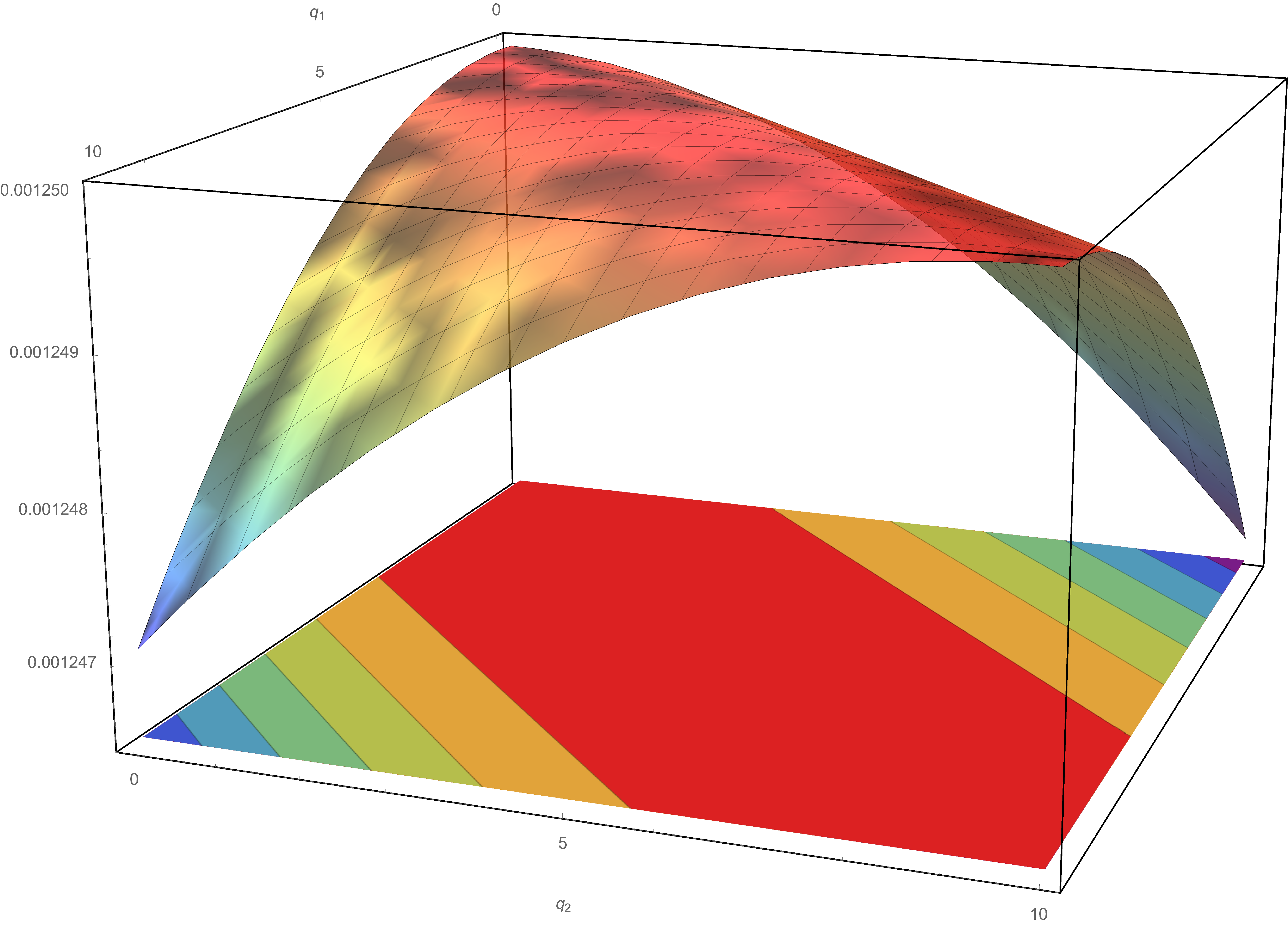}
\caption{We illustrate the difference between the spot price $p^1_t$ and the purchase price $\kappa^1_t(\Delta q)$ for asset one under varying values of $q_1$ and $q_2$ with $\Delta q = 1$. The value of $\beta = 1/100$. As $\beta$ decreases, the difference $\kappa^1_t(\Delta q) - p^1_t \to 0$. Thus ensuring \cref{eqn:BetaReq1,eqn:BetaReq2}}.
\label{fig:PriceDiff}
\end{figure}
Using the prior result, it is straightforward to see that if $D \subset\mathbb{R}^n$ is a simply connected closed and bounded set and $\chi_D(\mathbf{x})$ is its characteristic function, then if $\epsilon > 0$, there is a market $M$ with agents $\mathcal{A} = \{a^1,\dots,a^N\}$ (for some possibly large N) so that:
\begin{equation}
    \int_D \left|\chi_D(\mathbf{x}) - \bar{p}^1(\mathbf{x})\right|^2 \, d\mathbf{x} < \epsilon.
    \label{eqn:Approximation2}
\end{equation}
To see this, choose a large but finite sample of points from $D$ and add to this an appropriately large sample of points near the boundary of $D$. Call this set $\mathbf{H}$ and apply an argument like the one for \cref{prop:Approximation} to construct the market. From this we conclude:
\begin{proposition} If $D$ is a finite union of simply connected closed and bounded subsets of $\mathbb{R}^n$ and $\epsilon > 0$, then there is a market $M$ and a finite (but large) set of agents so that \cref{eqn:Approximation2} holds. \hfill\qedsymbol
\label{prop:Disjoint}
\end{proposition}
We effectively illustrate \cref{prop:Disjoint} in \cref{sec:Inventor}.

\subsection{Convergence}
Let:
\begin{equation}
    \Omega^i_t = 
    \left\{
    \mathbf{x} \in \mathbb{R}^n : 
    \frac{1}{\kappa^{y_i}_t}\left(\pi^i_t\left(\mathbf{x},p_t^{y_i};\mathbf{h}^i;\bm{\theta}^i,\alpha^i,w_p^i\right) - \kappa^{y_i}_t[\Delta q^{y_i}]\right) \geq \tau
    \right\}
\end{equation}
the following proposition provides a sufficient condition for the convergence of the market price to a single value.
\begin{proposition} Consider a market $M$ with agent set $\mathcal{A} = \{a^1,\dots,a^N\}$ and a fixed $\beta$, $\tau$. Given an input $\mathbf{x} \in \mathbb{R}^n$, if there is a time $t^*$ and an index set $I^* = \{i_1,\dots,i_k\} \subset \{1,\dots,N\}$ so that for all $t \geq t^*$:
\begin{equation}
    \mathbf{x} \in \bigcap_{i \in I^*} \Omega^i_t,
\end{equation}
and if $j \not\in I^*$, then $\mathbf{x} \not\in \Omega^j_t$, then the price $p^1_t$ converges to a fixed value.
\label{prop:Convergence}
\end{proposition}
\begin{proof} Suppose there is a $t^*$ and $I^* = \emptyset$. Then no agent purchases occur at time $t \geq t^*$ and the market price $p^1_t$ remains constant at the value $p^1_{t-1}$. If $I^*$ is not empty, then assume there are $r \geq 0$ agent indices who buy Asset 1 in $I^*$ and $s \geq 0$ agent indices who buy Asset 0 in $I^*$. Then for all time the spot price for Asset 1 is given by:
\begin{equation}
    p^1_t = \frac{\exp\left[\beta\left(q_{t-1}^1 + rt\right)\right]}{\exp\left[\beta\left(q_{t-1}^1 + rt\right)\right] + \exp\left[\beta\left(q_{t-1}^0 + st\right)\right]},
\end{equation}
because at all future times the agents in $I^*$ will purchase 1 unit of the appropriate asset. Taking the limit at $t \to \infty$ yields:
\begin{equation}
    \bar{p}^1 = \begin{cases}
    0 & \text{if $s > r$}\\
    1 & \text{if $s < r$}\\
    \frac{\exp\left(\beta q^1\right)}{\exp\left(\beta q^1\right) + \exp\left(\beta q^0\right)} & \text{if $r = s$}
    \end{cases}
\end{equation}
This completes the proof.
\end{proof}
We note that when each agent is given a finite bank account, then convergence of the market is ensured and the decision logic must be amended to include a test for sufficient funds. 

It is easy to construct an example in which the market does not converge to a fixed point. To see this, consider a market with a two dimensional feature space and two agents with $\mathbf{h}^1 = (0,0)$ and $\mathbf{h}^2 = (2,0)$. Let $\mathbf{x} = (1.02,0)$. Set $\tau = 0$, $\beta = 1/5$ and for $i=1,2$ set $\alpha^i = 3$, $w_p^i = 2$. If we assume both agents have $r^i_j = 1.015$ (i.e., agent geometry is circular), then this market will oscillate in price forever as illustrated in \cref{fig:OscillatingPrice}.
\begin{figure}[htbp]
\centering
\includegraphics[width=0.45\columnwidth]{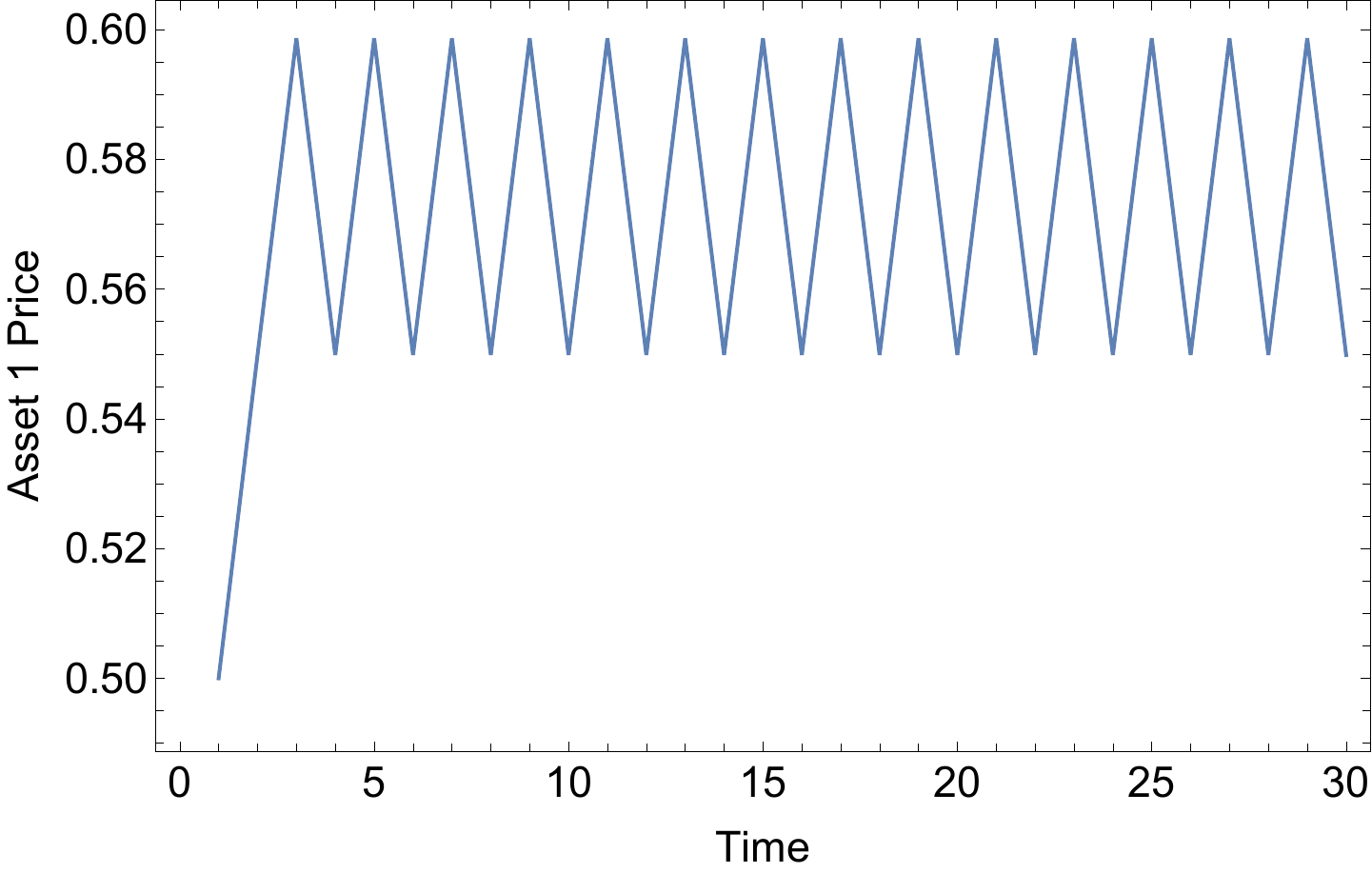}
\caption{An example of an oscillating market price for an input near the decision boundary.}
\label{fig:OscillatingPrice}
\end{figure}
The oscillation in the price is caused by the oscillation in the geometric structure of the sets $\Omega^1_t$ and $\Omega^2_t$. 
As the market price varies in time, each agent oscillates between determining the price is too high or low enough to purchase. Thus, when input information is close to a decision boundary we see that market prices may exhibit a limit cycle. Establishing sufficient conditions for the emergence of a limit cycle in the market is left to future work. However, as we have illustrated limit cycles will emerge when input (test) points are near multiple agent decision boundaries in feature space and thus can indicate indecision if the market is used as a machine learning model.

\section{Training Agents within a Market}\label{sec:Training}
In this section we discuss a practical implementation of the prediction market described above and detail a method to train such a market to approximate a data set. For practical purposes, we make three simplifying implementation changes:
\begin{enumerate}
    \item We assume time is finite. That is, the market will terminate after a fixed large time. 
    \item We assume all agents have a finite bank account.
    \item We assume that agents recurrently arrive at the market to buy assets with inter-arrival times governed by an exponential distribution. Thus not all agents interact with the market simultaneously.
\end{enumerate}
The third assumptions is made to increase the execution speed of the market and to ensure a sufficient number of training epochs can be executed in a reasonable amount of wall-clock time.

\subsection{Initialization}\label{subsec:Initialization}
Let each training data-point be denoted as $(\mathbf{x}^i, y^i)$, where $y^i$ denotes the output label. Let $m$ be the total number of training data-points. Define:
\begin{equation}
    \mathcal{C}_k = \left\{i \in \{1,\dots,m\} : y^i = k\right\},\\ 
\end{equation}
for $k = 0,1$. Training will proceed in batches. Assume a batch size of $b$, where $b < m$, to denote the number of data-points used to train the model in one pass. Thus, there will be a total of $\lceil m/b \rceil$ batches. A set of $n$ Agents are initialized for every data-point $\mathbf{x}_i$ in a batch $B_j$ where $i \in [1, b]$ and $j \in [1, \lceil m/b \rceil]]$. The agents are initialized as hyperspheres centered at $\mathbf{h}^i = \mathbf{x}^i$. To determine the initial radius, let
\begin{align*}
    r_1^i &= \argmin_{j \in \mathcal{C}_{y^i}} \left\lVert\mathbf{x}^i - \mathbf{x}^j\right\rVert^2\\
    r_2^i &= \frac{1}{2} \cdot \argmin_{j\in \mathcal{C}_{1-y^i}} \left\lVert\mathbf{x}^i - \mathbf{x}^j\right\rVert^2.
\end{align*}
These are the distances to the nearest point with similar classification and half the distance to the nearest point with opposite classification. Then set:
\begin{align*}
    p^i &= \max\left\{r_0, \min\left\{r_1^i, r_2^i\right\}\right\}\\
    q^i &= \max\left\{r_0', 2r_2^i\right\}
\end{align*}
where $r_0$ and $r_0'$ are default values. The radius of the hyper-sphere is initialized with:
\begin{displaymath}
    r^i \sim \mathcal{U}(p^i, q^i),
\end{displaymath}
where $\mathcal{U}$ is a uniform distribution. That is, we model each agent with an ellipsoid so all axial radii are initialized to $r^i$. The initial value for $w_p^i$ is chosen from a standard normal distribution for each $i$. Finally, if Agent $i$ is centered at $\mathbf{x}^i = \mathbf{h}^i$ with class $y^i$, then that agent will only purchase assets of Class $y^i$.

\subsection{Market Run}
Each market run is parameterized by an input feature vector $\mathbf{x}$ shared by all agents. This feature vector is used in agent purchase logic. Agents are initialized with a finite bank. During an execution of the market, each agent is seeded with an initial time it will interact with the market drawn from an exponential distribution. The next time of execution is set when the agent interacts with the market and uses the same exponential distribution. All agents have a common exponential distribution. Agents buy assets according to the decision logic discussed above and keep track of purchased assets and the price paid. There is a global clock that is updated to determine when agents participate. At market completion (after a fixed time has passed), agent profits and losses are calculated assuming assets that match the ground truth class $y^i$ are worth $1$ and other assets are valued at $0$.

\subsection{Evolutionary Algorithm}
The evolutionary algorithm defined below is used to identify parameters $\bm{\theta}^i$, $w^i_p$, and $\alpha$. For the purposes of this work, we assume that $\theta_0^i = 0$ is fixed for all $i$, we set $\beta = 1/100$ and $\tau = \frac{1}{10}$. Optimizing these parameters is a subject of future work.

\noindent\textbf{\underline{Evolutionary Algorithm}}\\
\noindent\textbf{Input:} Feature vectors $\mathbf{X}=\left\{\mathbf{x}^i\right\}_{i=1}^N$, ground truth labels $\left\{\mathbf{y}^i\right\}_{i=1}^N$
\begin{enumerate}
\item For each data point $\mathbf{x}^i$ with label $y^i$ create $n$ agents centered at $\mathbf{h}^i = \mathbf{x}^i$ and an initial random radius $r^i\sim \mathcal{U}(p^i, q^i)$, a random scale parameter $\alpha^i\sim \mathcal{U}(0.01, 5)$ and a random $w_p^i\sim \mathcal{N}(0, 1)$. Agents specialize in the purchase of shares of type $y^i$. 

\item Run $N$ markets one for each input $\mathbf{x}^i \in \mathbf{X}$. 

\item For each market, sort all agents into three groups (i) those that did not participate, (ii) those that made a profit and (iii) those that had a loss. 
\begin{enumerate}
\item For each center $\mathbf{h}^i$:
\begin{enumerate}
\item If no agent with center $\mathbf{h}^i$ participated, continue.
\item Among the agents who participated retain $l < n$ agents who had the highest profit (or lowest loss). 
\item Delete the $n-l$ under-performing agents.
\item Create $n-l$ new agents from the agent pool centered at $\mathbf{h}^i$ using mutation and crossover of the parameters $\alpha$, $\bm{\theta}$ and $w_p$. Specifically, mutation is carried out as follows:
\begin{enumerate}
    \item Compute 
    \begin{equation}\sigma = 2 \sqrt{\frac{1}{n}\Sigma_{i=1}^m \left(y^i - \bar{p}^{y^i}\right)^2}.
    \label{eqn:VarianceControl}
    \end{equation}
    \item Update $r^i \leftarrow r^i + \sigma \cdot \mathcal{U}(p^i - r^i, q^i - r^i)$. 
    \item Update $w_p^i \leftarrow w_p^i + \sigma \cdot \mathcal{N}(0,1)$.
\end{enumerate}
\end{enumerate}
\end{enumerate}
\item Goto 2. This process is repeated for $g$ generations.
\end{enumerate}
Because each agent is modeled by an ellipsoid with a finite volume, not every agent will participate in every market. In particular, if $\mathbf{x}^i \not\in\Omega^i_{t}$ for any time $t$, then Agent $i$ will not participate. Of those agents that do participate in a given market, those that are most successful are preserved and replicate with mutation and crossover. The mutation rate is controlled by the current root mean-square error of the approximation. As this value decreases, the mutation decreases. 

\section{Experimental Results}\label{sec:Results}
This section discusses the results obtained by the application of the proposed market model on standard datasets such as IRIS Dataset and Heart Disease Dataset. We also apply the model to perform the record linkage task of disambiguating inventor records from the USPTO PatentsView database as a real-world application usecase. For all the experiments, we have chosen $n = 5$ (agent replicants), $l = 3$ (retained agents) and $g = 20$ (generations).

\subsection{IRIS Dataset}
We study the standard IRIS data set \cite{kholerdi2018enhancement}, which consists of features describing three species of iris plants - Iris setosa, Iris virginica and Iris versicolor. The data set contains 50 instances of feature vectors from each class. It is known that Iris Setosa is \textit{linearly separable} from the other two classes. However, Iris Versicolor and Iris Virginica are \textit{not linearly separable} from each other. We use four attributes, length and width of sepals and petals, to classify an instance into one of the three classes.

The proposed market model is generalized to be a binary classifier. However, the dataset consists of three classes. Therefore, we take the union of two classes and train the model on the one-against-two binary classification problem. We used a train-test split of 75:25.
\paragraph{Union of Iris Setosa and Iris Versicolor}
We combined the two classes, Iris Setosa and Iris Versicolor, and represented them as Class 0. Class 1 was composed of data from to Iris Virginica. A test accuracy of 94.6\% was obtained in this case. A detailed analysis is shown in \cref{setosa_versicolor_combined_perf_metrics}. 
\begin{table}[!htb]
\centering
\begin{tabular}{l|lll}
\textbf{Class}                          & \textbf{Precision} & \textbf{Recall} & \textbf{F1-Score} \\
\hline
Class 0 (Setosa/Versicolor) & 1.00               & 0.91            & 0.95              \\
Class 1 (Virginica)                 & 0.88               & 1.00            & 0.95             
\end{tabular}
\caption{Shows the test performance when instances of Iris Setosa and Iris Versicolor are combined together as Class0.}
\label{setosa_versicolor_combined_perf_metrics}
\end{table}

\paragraph{Union of Iris Setosa and Iris Virginica}
We combined the instances of Iris Setosa and Iris Virginica as Class 0. Class 1 contains data from Iris Versicolor. A test accuracy of 97.29\% was observed and \cref{setosa_virginica_combined_perf_metrics} shows a detailed analysis.
\begin{table}[!htb]
\centering
\begin{tabular}{l|lll}
\textbf{Class}                         & \textbf{Precision} & \textbf{Recall} & \textbf{F1-Score} \\
\hline
Class 0 (Setosa/Virginica) & 0.96               & 1.00            & 0.98              \\
Class 1 (Versicolor)               & 1.00               & 0.93            & 0.96             
\end{tabular}
\caption{Shows the test performance when instances of Iris Setosa and Iris Virginca are combined together as Class0.}
\label{setosa_virginica_combined_perf_metrics}
\end{table}

\paragraph{Union of Iris Versicolor and Iris Virginica}
We combined the instances of Iris Versicolor and Iris Virginica as Class 1. Class 0 consists of data from Iris Setosa. A test accuracy of 100.0\% was observed \cref{versicolor_virginica_combined_perf_metrics} shows a detailed analysis. We have to note that instances of Iris Setosa are linearly separable from the other two classes and thus, the model is able to separate the two classes with 100\% accuracy in this case.
\begin{table}[!htb]
\centering
\begin{tabular}{l|lll}
\textbf{Class}                             & \textbf{Precision} & \textbf{Recall} & \textbf{F1-Score} \\
\hline
Class 0 (Setosa)                       & 1.00               & 1.00            & 1.00              \\
Class 1 (Versicolor/Virginica) & 1.00               & 1.00            & 1.00             
\end{tabular}
\caption{Shows the test performance when instances of Iris Setosa and Iris Virginca are combined together as Class0.}
\label{versicolor_virginica_combined_perf_metrics}
\end{table}

\subsection{Heart Disease Dataset}
This is a publicly available dataset \cite{Dua:2019} provided by UCI. There are four databases available for use within the dataset. Published experiments in Machine Learning use the Cleveland database with a maximum of 14 of the 76 available attributes which are known to be considerably linked to heart disease. We use the following 14 numerical attributes to train the market model to classify patients to one of the targets; presence of heart disease, no heart disease. 
\begin{enumerate}
    \item Age
    \item Sex: male, female
    \item Chest pain type: typical angina (angina), atypical angina (abnang), non-anginal pain (notang), asymptomatic (asymp)
    \item Trestbps: resting blood pressure on admission
    \item Chol: serum cholestrol
    \item Fbs: indicates whether fasting blood sugar is greater than 120 mg/dl
    \item Restecg: normal(norm), abnormal(abn): ST-T wave abnormality, ventricular hypertrophy (hyp) 
    \item Thalach: maximum heart rate achieved
    \item Exang: exercise induced angina
    \item Oldpeak: ST depression induced by exercise relative to rest
    \item Slope: upsloping, flat, downsloping: the slope characteristics of the peak exercise ST segment
    \item Ca: number of fluoroscopy colored major vessels
    \item Thal: normal, fixed defect, reversible defect - the heart status
    \item Class/target label
\end{enumerate}
The data set has a total of 303 data points. To evaluate the performance of the market (M), we split the data using an 80\%-20\% ratio. The market was tested on 20\% of the randomly sampled data. A total of 60 data points were used for testing the model. For one of the randomly chosen split, we obtained a test accuracy of 86.66\%. The confusion matrix associated with the test data is shown in \cref{fig:cf_matrix_heart_disease}.      

\begin{figure}[htbp]
\centering
\includegraphics[width=0.45\columnwidth]{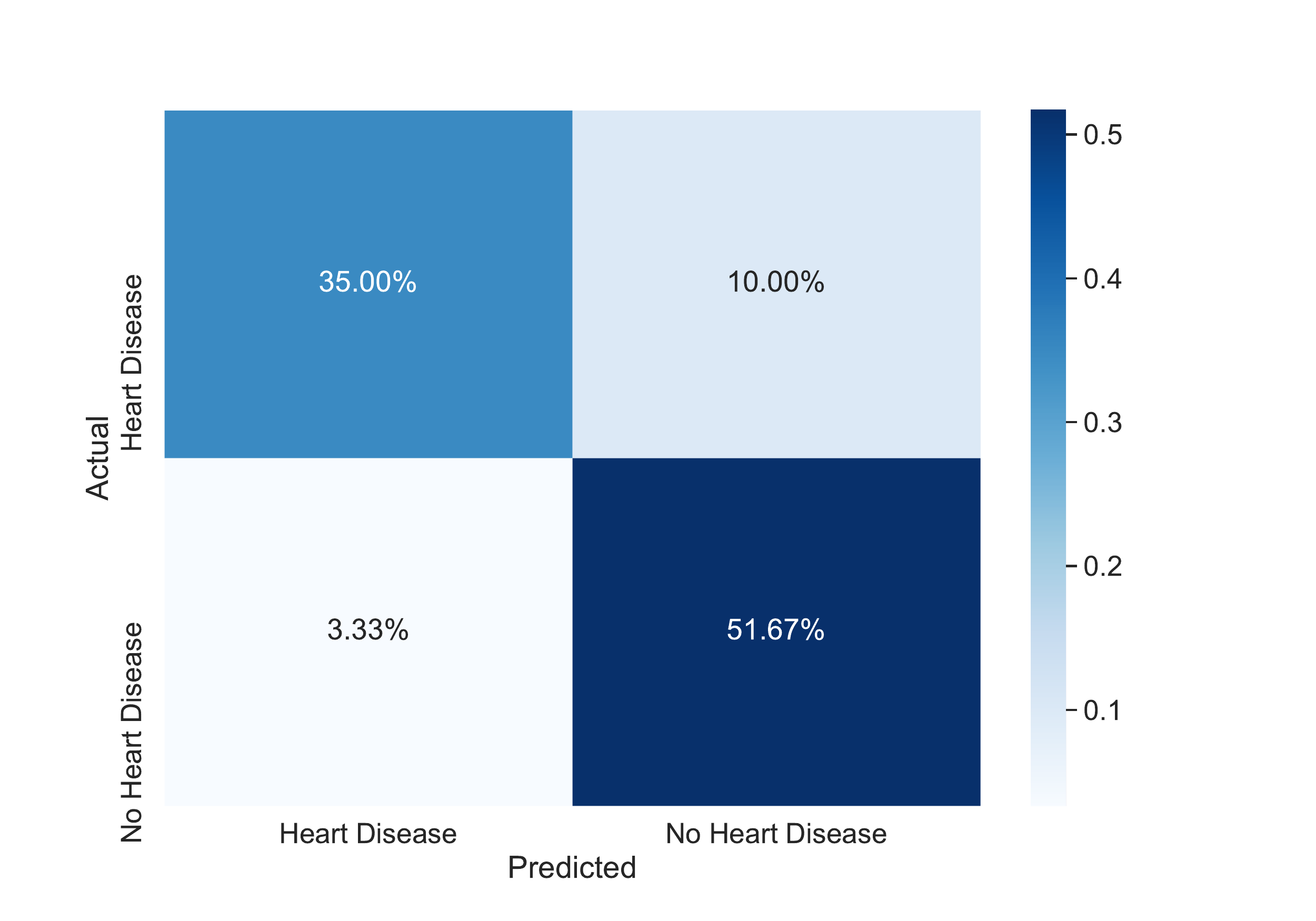}
\caption{Confusion Matrix obtained for the test set.}
\label{fig:cf_matrix_heart_disease}
\end{figure}

The obtained results from the market model are compared with the output obtained from a Random Forest (RF) classifier for the same split. The RF classifier obtained a test accuracy of 96.66\%. \cref{tab:f1_score_rf_market} compares the F1-Score obtained for both the models. We see that Random Forest outperformed the market in this case.

\begin{table}[!htb]
\centering
\begin{tabular}{|c|c|c|}
\hline
\multicolumn{1}{|l|}{{\color[HTML]{000000} \textbf{Model}}} & \multicolumn{1}{l|}{{\color[HTML]{000000} \textbf{No Heart Disease (\%)}}} & \multicolumn{1}{l|}{{\color[HTML]{000000} \textbf{Heart Disease (\%)}}} \\ \hline
\textbf{RF}                                                 & 96                                                                         & 97                                                                      \\ \hline
\textbf{M}                                                  & 84                                                                         & 89                                                                      \\ \hline
\end{tabular}
\caption{Shows the F1 scores for each class for both the classifiers.}
\label{tab:f1_score_rf_market}
\end{table}


\begin{figure}[htbp]
\centering
\includegraphics[width=0.45\columnwidth]{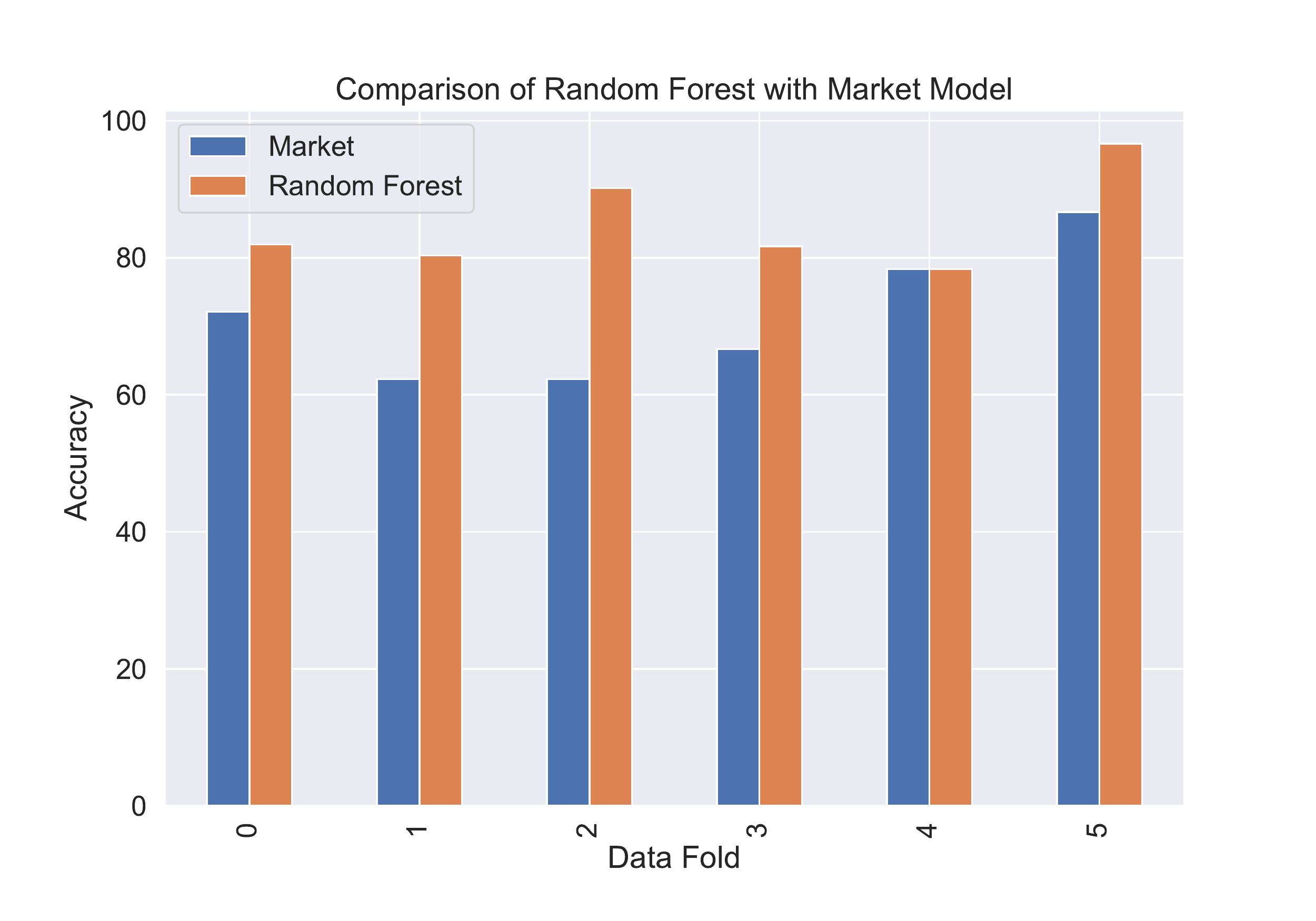}
\caption{Comparison of the Market Model with Random Forest Classifier for six different data splits.}
\label{fig:rf_market_comparison}
\end{figure}

\begin{figure}[htbp]
\centering
\subfloat[Class - Heart Disease]{\includegraphics[width=0.45\columnwidth]{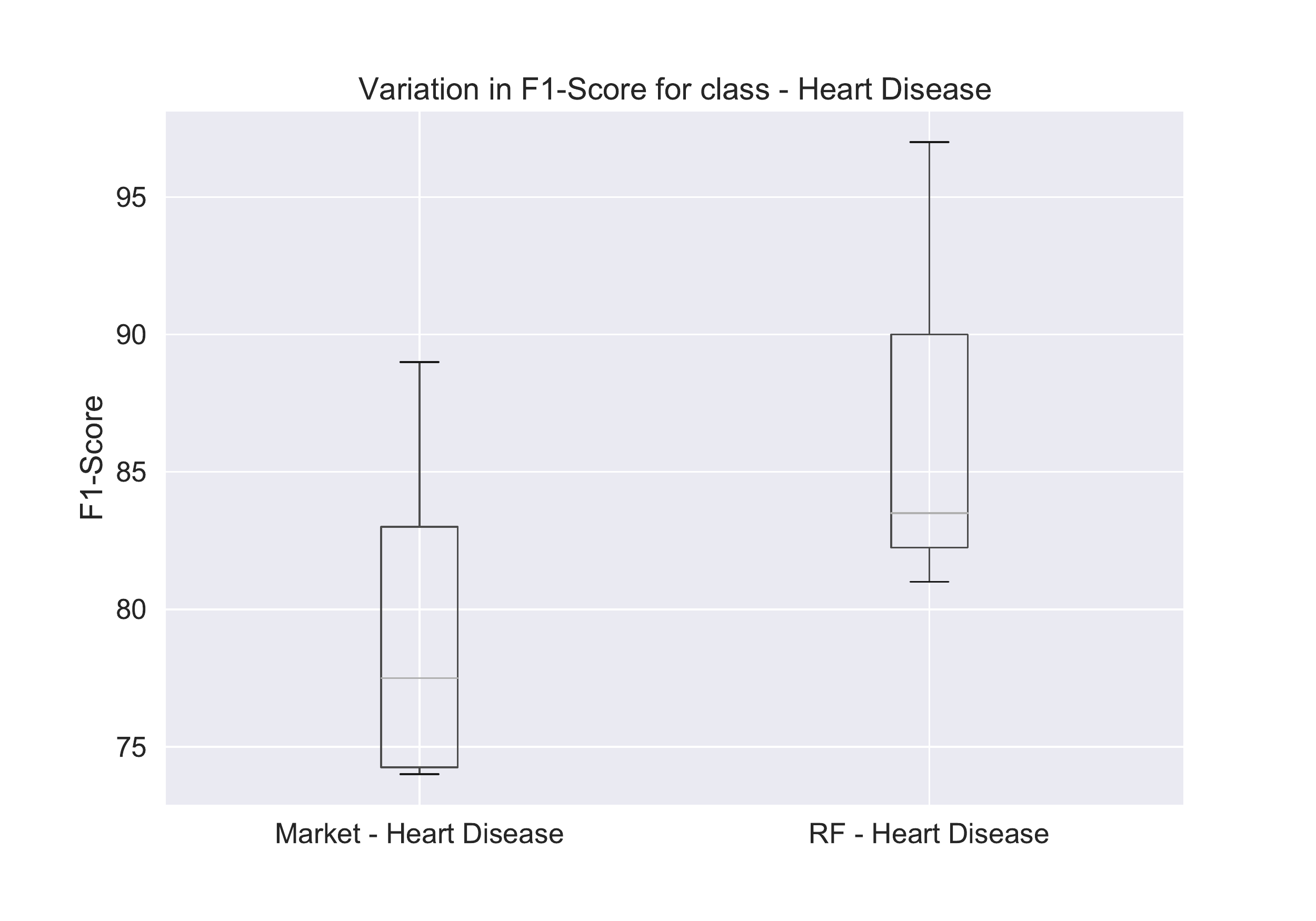}} \quad
\subfloat[Class - No Heart Disease]{\includegraphics[width=0.45\columnwidth]{Figures/F1_score_heart_disease.pdf}}
\caption{Variation in F1-Scores for both models in each class.}
\label{fig:f1_score_comparison}
\end{figure}

\begin{figure}[htbp]
\centering
\includegraphics[width=0.45\columnwidth]{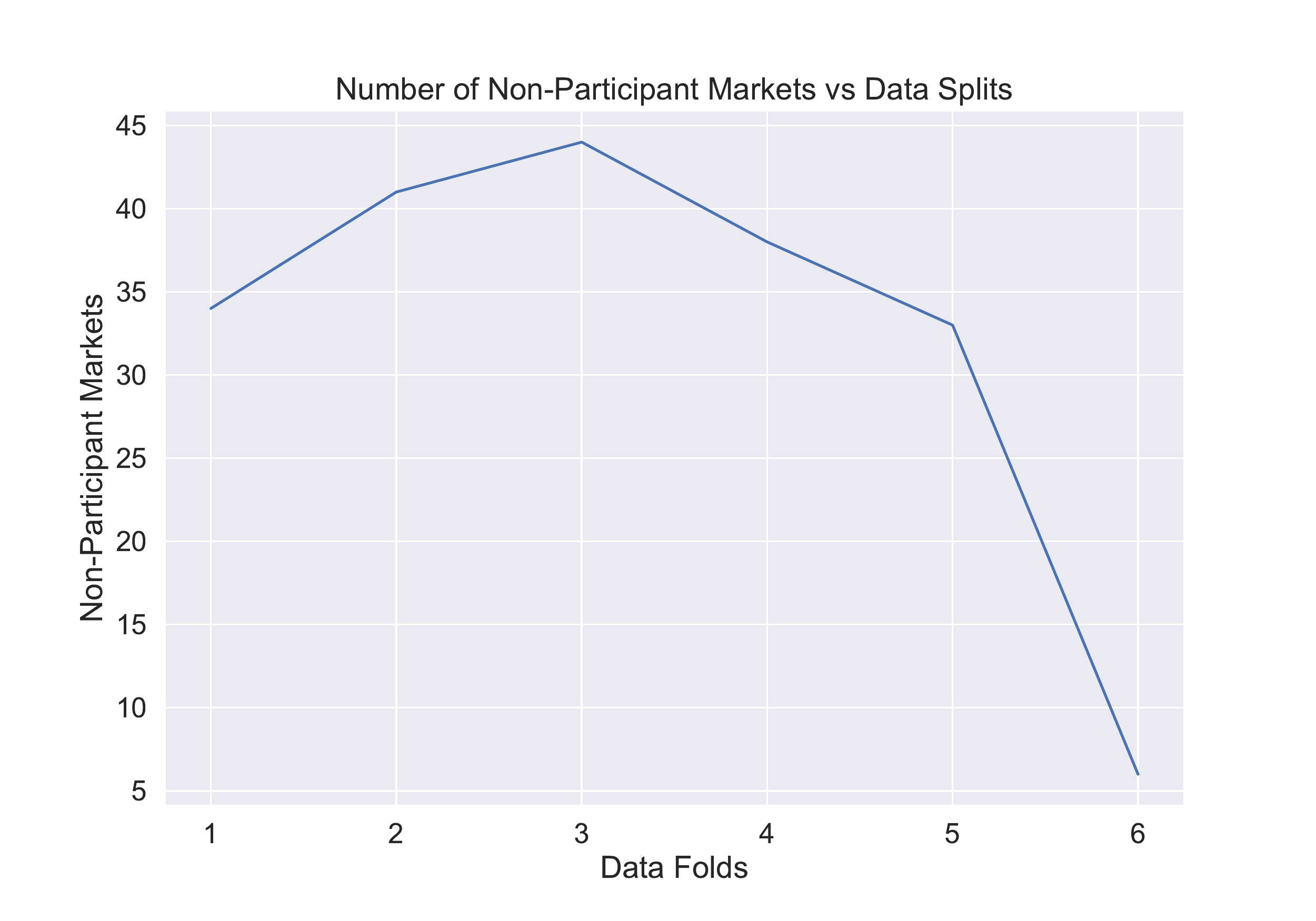}
\caption{Shows the number of markets with No-Agent Participation for six different data splits.}
\label{fig:non_participants}
\end{figure}

To measure the sensitivity of the model with respect to the variation in inputs, we performed the experiment with six randomly sampled data splits. A train-test ratio of 80\%-20\% was retained for all the splits. \cref{fig:rf_market_comparison} shows the comparison of the Market model with RF Classifier. We observe variations in the two models with respect to changing input data. The RF classifier outperformed the market in five out of six cases. Market performance was comparable to RF classifier for the fifth split. \cref{fig:f1_score_comparison} shows the box plot of F1-scores associated with each class for the two models.

The lower performance of the market can be attributed to a lack of generalization using the underlying geometry. The use of simple geometric agents allows us to quantify this. \cref{fig:non_participants} shows the number of markets with no agent participation for various data splits. We observed that accuracy increased with increase in agent participation across all markets. The highest obtained accuracy of 86.66\%, as seen in \cref{fig:rf_market_comparison}, had only six markets with no agent participation. For future work, we will include agents whose decision logic is characterized by either ellipsoids or a convex cone, which generalizes a hyperplane separator but also remains interpretable. The challenge in this will be to alter the evolutionary algorithm for account for agents with multiple geometries.



\subsection{Inventor Disambiguation}\label{sec:Inventor}

We curate a subset of the publicly available database that was released as part of the PatentsView\footnote{{https://www.patentsview.org/download/}} Inventor Disambiguation Workshop. A random sample of labeled inventor records from the database containing 346 patent records of 74 distinct inventors is selected. We then build a data set based on pairwise similarity vectors of records in the sample to train the prediction market classifier. We have three distinct types of pairs for the inventor/patent records in the data set which yields a total of 2646 patent record pairs:
\begin{enumerate}
    \item Positive pairs: We leverage the labeled records from the database relating to distinct inventors to create inventor clusters where all possible pairs within a cluster are assigned a label 1 indicating that they are the same person. 
    \item Similar negative pairs: For each inventor cluster in the sample, we retrieve patent records from the database such that the first name and the last name of the inventor are an exact match, but they are a different person; i.e., the record does not belong to this inventor cluster. These are prime candidates that make disambiguation necessary and have a label 0.
    \item Random negative pairs: We generate random pairs by using candidates that belong to different clusters that are assigned a label 0. 
\end{enumerate}

We compute the similarity and distance measures outlined in \cref{tab:disam_features} for each patent record pair for the respective inventor features in the sample.

\begin{table}[htbp]
\begin{center}\small
\begin{tabular}{c|c|c} 
\toprule
{\bf Type} & {\bf Measure} & {\bf Feature}\\
\midrule
\multirow{6}{*}{Token} & \multirow{6}{*}{Cosine, Jaccard Similarity} &  Title \\
                                                                     & & Section \\
                                                                     & & Subsection \\
                                                                     & & Group \\
                                                                     & & Sub-group \\ 
                                                                     & & Organization \\ 
                                                                    \midrule
\multirow{4}{*}{String} & \multirow{4}{*}{Jaro-Winkler, Soundex} &  First Name \\
                                                                     & & Last Name \\
                                                                     & & City \\
                                                                     & & State \\
                                                                    \midrule
 Geographic distance  & Haversine distance & Latitude $|$ Longitude \\              \bottomrule
 
\end{tabular}
\caption{USPTO Inventor Disambiguation pairwise feature classes}
\label{tab:disam_features}
\end{center}
\end{table}


 \cref{fig:disambiguation_dataset_plot_perplexity} offers a visualization of the data set in two dimensions obtained using t-SNE \cite{Maaten2008VisualizingDU} for nonlinear dimensional reduction. It showcases a general sense of the topology intrinsic to the data. Cluster sizes may not mean anything, nor does the distance between identified clusters or within them as discussed in \cite{Wattenberg2016HowTU}. However, we see that there exists a non-linear, complex separation in the overall topological structure, consistent across different perplexity settings. The two classes have cases where there are clearly identifiable clusters and some less so with overlapping instances that require complex decision boundaries. The occurrences of outlier instances of a class in clusters predominantly composed of instances of the other class is of particular interest. These cases are a great test to validate the performance of the synthetic prediction market classifier, which can utilize the subtle variances in local geometry to differentiate between the two classes. This also illustrates \cref{prop:Disjoint}. This data set is an ideal candidate to showcase the market's ability to distinguish non-linearly separable data. We hypothesized that the agent initialization as discussed in \cref{subsec:Initialization} would enable the market model to perform well by offering a good covering of the data set and is experimentally confirmed as discussed in subsequent results.

\begin{figure}[htbp]
\centering
\includegraphics[width=0.95\columnwidth]{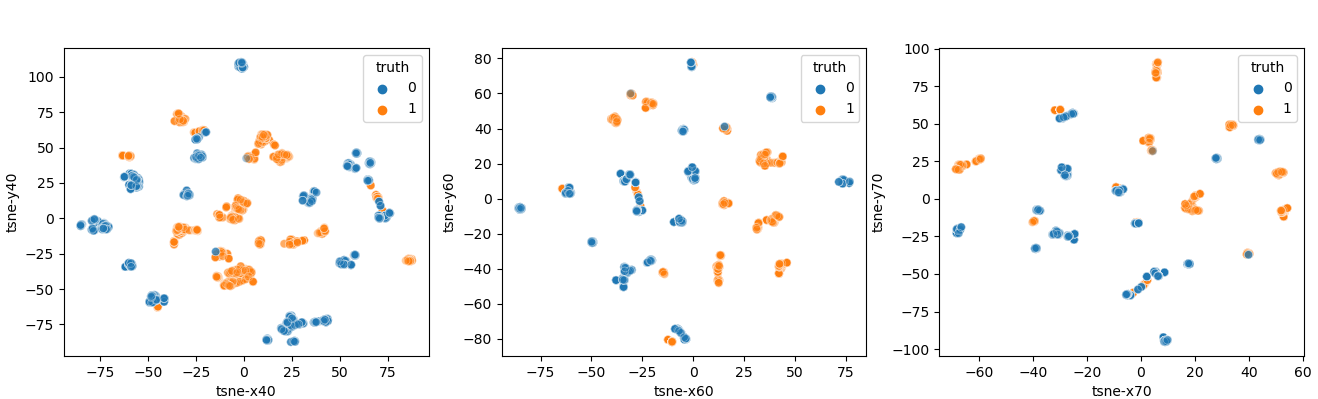}
\caption{Two dimensional t-SNE plots of the pairwise feature space for different perplexity configurations. Left to right: 40, 60, and 70.}
\label{fig:disambiguation_dataset_plot_perplexity}
\end{figure}



To evaluate the performance of the market, we split the data using an 80\%-20\% ratio. \cref{tab:disam_comparison} compares the performance of our model against a classic machine learning model (Random Forest Classifier). We see that both models perform very well for this dataset and obtain a classification accuracy of 99.81\%.

\begin{table}
\begin{center}\small
\begin{tabular}{c|c|c|c}
    \toprule
    {\bf Model} & {\bf Precision} & {\bf Recall} & {\bf F1-Score} \\
    \midrule
    RF & 0.996 & 1.0 & 0.998 \\
    \midrule
    M & 0.996 & 1.0 & 0.998 \\
     \bottomrule
\end{tabular}
\end{center}
\caption{F1 score for each model for USPTO inventor disambiguation.}
\label{tab:disam_comparison}
\end{table}

\section{Conclusions and Future Directions}\label{sec:Conclusions}
In this paper we study a specific class of synthetic binary prediction markets in which agent decision logic is specified by a convex semi-algebraic set. We showed that these prediction markets satisfy certain universal approximation properties and gave sufficient conditions for the market to converge to a final set of asset prices. We also showed that these markets can enter limit cycles, which indicate the the conditioning data may be near a decision boundary. We provided an evolutionary algorithm for training such a market on a given data set and illustrated this process on three example data sets. While the market under-performed the best-in-class random forest algorithms for some data sets, we were able to show consistent or equal performance to the random forest method in all tests. In addition, we use the underlying geometric structures to infer the reason for the under-performance and devised an approach to mitigate this in future work. 

For future work, we will introduce agents whose decision rules are characterized by convex cones. These agents will work along side the existing agents (who use ellipsoidal regions) to characterize data sets, thus improving generalization. We will also study the possible dynamics of these markets and determine whether a more robust stability theorem can be proven.

\section*{Acknowledgement} 
Portions of this work were sponsored by the DARPA SCORE Program (Cooperative Agreement W911NF-19-2-0272.)



\bibliographystyle{unsrt} 
\bibliography{DARPAMarketRefs}

\begin{thebibliography}{10}

\bibitem{hanson1990market}
Robin Hanson.
\newblock Market-based foresight-a proposal.
\newblock {\em Foresight Update}, 10(1):3, 1990.

\bibitem{hanson1991more}
R~Hanson.
\newblock More market-based foresight.
\newblock {\em Foresight Update}, 11(11), 1991.

\bibitem{hanson1995could}
Robin Hanson.
\newblock Could gambling save science? encouraging an honest consensus.
\newblock 1995.

\bibitem{ray1997idea}
Russ Ray.
\newblock Idea futures: Gambling on science.
\newblock {\em The Futurist}, 31(1):25, 1997.

\bibitem{wolfers2004prediction}
Justin Wolfers and Eric Zitzewitz.
\newblock Prediction markets.
\newblock {\em Journal of economic perspectives}, 18(2):107--126, 2004.

\bibitem{servan2004prediction}
Emile Servan-Schreiber, Justin Wolfers, David~M Pennock, and Brian Galebach.
\newblock Prediction markets: Does money matter?
\newblock {\em Electronic markets}, 14(3):243--251, 2004.

\bibitem{manski2006interpreting}
Charles~F Manski.
\newblock Interpreting the predictions of prediction markets.
\newblock {\em economics letters}, 91(3):425--429, 2006.

\bibitem{berg2003prediction}
Joyce~E Berg and Thomas~A Rietz.
\newblock Prediction markets as decision support systems.
\newblock {\em Information systems frontiers}, 5(1):79--93, 2003.

\bibitem{wolfers2006prediction}
Justin Wolfers and Eric Zitzewitz.
\newblock Prediction markets in theory and practice.
\newblock Technical report, national bureau of economic research, 2006.

\bibitem{dai2020wisdom}
Min Dai, Yanwei Jia, and Steven Kou.
\newblock The wisdom of the crowd and prediction markets.
\newblock {\em Journal of Econometrics}, 2020.

\bibitem{chakraborty2016trading}
Mithun Chakraborty and Sanmay Das.
\newblock Trading on a rigged game: Outcome manipulation in prediction markets.
\newblock In {\em IJCAI}, pages 158--164, 2016.

\bibitem{tziralis2007prediction}
George Tziralis and Ilias Tatsiopoulos.
\newblock Prediction markets: An extended literature review.
\newblock {\em The journal of prediction markets}, 1(1):75--91, 2007.

\bibitem{berg1997makes}
Joyce Berg, Robert Forsythe, and Thomas Rietz.
\newblock What makes markets predict well? evidence from the iowa electronic
  markets.
\newblock In {\em Understanding Strategic Interaction}, pages 444--463.
  Springer, 1997.

\bibitem{thaler1988anomalies}
Richard~H Thaler and William~T Ziemba.
\newblock Anomalies: Parimutuel betting markets: Racetracks and lotteries.
\newblock {\em Journal of Economic perspectives}, 2(2):161--174, 1988.

\bibitem{wolfers2006interpreting}
Justin Wolfers and Eric Zitzewitz.
\newblock Interpreting prediction market prices as probabilities.
\newblock Technical report, National Bureau of Economic Research, 2006.

\bibitem{polgreen2007use}
Philip~M Polgreen, Forrest~D Nelson, George~R Neumann, and Robert~A Weinstein.
\newblock Use of prediction markets to forecast infectious disease activity.
\newblock {\em Clinical Infectious Diseases}, 44(2):272--279, 2007.

\bibitem{almenberg2009experiment}
Johan Almenberg, Ken Kittlitz, and Thomas Pfeiffer.
\newblock An experiment on prediction markets in science.
\newblock {\em PLoS One}, 4(12):e8500, 2009.

\bibitem{dreber2015using}
Anna Dreber, Thomas Pfeiffer, Johan Almenberg, Siri Isaksson, Brad Wilson,
  Yiling Chen, Brian~A Nosek, and Magnus Johannesson.
\newblock Using prediction markets to estimate the reproducibility of
  scientific research.
\newblock {\em Proceedings of the National Academy of Sciences},
  112(50):15343--15347, 2015.

\bibitem{cowgill2009using}
Bo~Cowgill, Justin Wolfers, and Eric Zitzewitz.
\newblock Using prediction markets to track information flows: Evidence from
  google.
\newblock In {\em amma}, page~3, 2009.

\bibitem{gillen2012information}
Benjamin~J Gillen, Charles~R Plott, and Matthew Shum.
\newblock Information aggregation mechanisms in the field: Sales forecasting
  inside intel.
\newblock Technical report, Working paper, 2012.

\bibitem{smith2003constructivist}
Vernon~L Smith.
\newblock Constructivist and ecological rationality in economics.
\newblock {\em American economic review}, 93(3):465--508, 2003.

\bibitem{arrow2008promise}
Kenneth~J Arrow, Robert Forsythe, Michael Gorham, Robert Hahn, Robin Hanson,
  John~O Ledyard, Saul Levmore, Robert Litan, Paul Milgrom, Forrest~D Nelson,
  et~al.
\newblock The promise of prediction markets.
\newblock {\em Science}, 320(5878):877, 2008.

\bibitem{goel2010prediction}
Sharad Goel, Daniel~M Reeves, Duncan~J Watts, and David~M Pennock.
\newblock Prediction without markets.
\newblock In {\em Proceedings of the 11th ACM conference on Electronic
  commerce}, pages 357--366, 2010.

\bibitem{expectations1961theory}
Rational Expectations.
\newblock the theory of price movements.
\newblock {\em Econometrica}, 29(3):315--35, 1961.

\bibitem{tetlock2008liquidity}
Paul~C Tetlock.
\newblock Liquidity and prediction market efficiency.
\newblock {\em Available at SSRN 929916}, 2008.

\bibitem{hanson2004manipulators}
Robin Hanson and Ryan Oprea.
\newblock Manipulators increase information market accuracy.
\newblock {\em George Mason University}, 2004.

\bibitem{tetlock2007optimal}
Paul~C Tetlock and Robert~W Hahn.
\newblock Optimal liquidity provision for decision makers.
\newblock {\em AEI-Brookings Joint Center Working Paper}, (06-18), 2007.

\bibitem{sunstein2006infotopia}
Cass~R Sunstein.
\newblock {\em Infotopia: How many minds produce knowledge}.
\newblock Oxford University Press, 2006.

\bibitem{ottaviani2007outcome}
Marco Ottaviani and Peter~Norman S{\o}rensen.
\newblock Outcome manipulation in corporate prediction markets.
\newblock {\em Journal of the European Economic Association}, 5(2-3):554--563,
  2007.

\bibitem{pennock2001real}
David~M Pennock, Steve Lawrence, C~Lee Giles, Finn~Arup Nielsen, et~al.
\newblock The real power of artificial markets.
\newblock {\em Science}, 291(5506):987--988, 2001.

\bibitem{rosenbloom2006statistical}
Earl~S Rosenbloom and William Notz.
\newblock Statistical tests of real-money versus play-money prediction markets.
\newblock {\em Electronic Markets}, 16(1):63--69, 2006.

\bibitem{gruca2008incentive}
Thomas~S Gruca, Joyce~E Berg, and Michael Cipriano.
\newblock Incentive and accuracy issues in movie prediction markets.
\newblock {\em The Journal of Prediction Markets}, 2(1):29--43, 2008.

\bibitem{Augur}
Augur: Your global, no-limit betting platform.
\newblock \url{https://www.augur.net}.
\newblock Accessed: 2020-10-04.

\bibitem{Gnosis}
Gnosis: Redistribute the future.
\newblock \url{https://gnosis.io}.
\newblock Accessed: 2020-10-04.

\bibitem{Stox}
Stox: The blockchain prediction markets platform.
\newblock \url{https://www.stox.com/}.
\newblock Accessed: 2020-10-04.

\bibitem{clark2014decentralizing}
Jeremy Clark, Joseph Bonneau, Edward~W Felten, Joshua~A Kroll, Andrew Miller,
  and Arvind Narayanan.
\newblock On decentralizing prediction markets and order books.
\newblock In {\em Workshop on the Economics of Information Security, State
  College, Pennsylvania}, volume 188, 2014.

\bibitem{heilman2016blindly}
Ethan Heilman, Foteini Baldimtsi, and Sharon Goldberg.
\newblock Blindly signed contracts: Anonymous on-blockchain and off-blockchain
  bitcoin transactions.
\newblock In {\em International conference on financial cryptography and data
  security}, pages 43--60. Springer, 2016.

\bibitem{bentov2017decentralized}
Iddo Bentov, Alex Mizrahi, and Meni Rosenfeld.
\newblock Decentralized prediction market without arbiters.
\newblock In {\em International Conference on Financial Cryptography and Data
  Security}, pages 199--217. Springer, 2017.

\bibitem{peterson2015augur}
Jack Peterson and Joseph Krug.
\newblock Augur: a decentralized, open-source platform for prediction markets.
\newblock {\em arXiv preprint arXiv:1501.01042}, 2015.

\bibitem{subramanian2017decentralized}
Hemang Subramanian.
\newblock Decentralized blockchain-based electronic marketplaces.
\newblock {\em Communications of the ACM}, 61(1):78--84, 2017.

\bibitem{wang2018preliminary}
Shuai Wang, Xiaochun Ni, Yong Yuan, Fei-Yue Wang, Xiao Wang, and Liwei Ouyang.
\newblock A preliminary research of prediction markets based on blockchain
  powered smart contracts.
\newblock In {\em 2018 IEEE International Conference on Internet of Things
  (iThings) and IEEE Green Computing and Communications (GreenCom) and IEEE
  Cyber, Physical and Social Computing (CPSCom) and IEEE Smart Data
  (SmartData)}, pages 1287--1293. IEEE, 2018.

\bibitem{barbu2012introduction}
Adrian Barbu and Nathan Lay.
\newblock An introduction to artificial prediction markets for classification.
\newblock {\em The Journal of Machine Learning Research}, 13(1):2177--2204,
  2012.

\bibitem{barbu2013artificial}
Adrian Barbu and Nathan Lay.
\newblock Artificial prediction markets for lymph node detection.
\newblock In {\em 2013 E-Health and Bioengineering Conference (EHB)}, pages
  1--7. IEEE, 2013.

\bibitem{jahedpari2014artificial}
Fatemeh Jahedpari, Julian Padget, Marina De~Vos, and Benjamin Hirsch.
\newblock Artificial prediction markets as a tool for syndromic surveillance.
\newblock {\em Crowd Intelligence: Foundations, Methods and Practices}, 2014.

\bibitem{chen2008complexity}
Yiling Chen, Lance Fortnow, Nicolas Lambert, David~M Pennock, and Jennifer
  Wortman.
\newblock Complexity of combinatorial market makers.
\newblock In {\em Proceedings of the 9th ACM conference on Electronic
  commerce}, pages 190--199, 2008.

\bibitem{chen2010new}
Yiling Chen and Jennifer~Wortman Vaughan.
\newblock A new understanding of prediction markets via no-regret learning.
\newblock In {\em Proceedings of the 11th ACM conference on Electronic
  commerce}, pages 189--198, 2010.

\bibitem{abernethy2011optimization}
Jacob Abernethy, Yiling Chen, and Jennifer Wortman~Vaughan.
\newblock An optimization-based framework for automated market-making.
\newblock In {\em Proceedings of the 12th ACM conference on Electronic
  commerce}, pages 297--306, 2011.

\bibitem{lay2012artificial}
Nathan Lay and Adrian Barbu.
\newblock The artificial regression market.
\newblock {\em arXiv preprint arXiv:1204.4154}, 2012.

\bibitem{storkey2011machine}
Amos Storkey.
\newblock Machine learning markets.
\newblock In {\em Proceedings of the Fourteenth International Conference on
  Artificial Intelligence and Statistics}, pages 716--724, 2011.

\bibitem{storkey2012isoelastic}
Amos Storkey, Jono Millin, and Krzysztof Geras.
\newblock Isoelastic agents and wealth updates in machine learning markets.
\newblock {\em arXiv preprint arXiv:1206.6443}, 2012.

\bibitem{hu2014multi}
Jinli Hu and Amos Storkey.
\newblock Multi-period trading prediction markets with connections to machine
  learning.
\newblock In {\em International Conference on Machine Learning}, pages
  1773--1781, 2014.

\bibitem{jahedpari2017online}
Fatemeh Jahedpari, Talal Rahwan, Sattar Hashemi, Tomasz~P Michalak, Marina
  De~Vos, Julian Padget, and Wei~Lee Woon.
\newblock Online prediction via continuous artificial prediction markets.
\newblock {\em IEEE Intelligent Systems}, 32(1):61--68, 2017.

\bibitem{hanson2007logarithmic}
Robin Hanson.
\newblock Logarithmic markets scoring rules for modular combinatorial
  information aggregation.
\newblock {\em The Journal of Prediction Markets}, 1(1):3--15, 2007.

\bibitem{lekwijit2018optimizing}
Suparerk Lekwijit and Daricha Sutivong.
\newblock Optimizing the liquidity parameter of logarithmic market scoring
  rules prediction markets.
\newblock {\em Journal of Modelling in Management}, 2018.

\bibitem{kholerdi2018enhancement}
Hedyeh~A Kholerdi, Nima TaheriNejad, and Axel Jantsch.
\newblock Enhancement of classification of small data sets using
  self-awareness—an iris flower case-study.
\newblock In {\em 2018 IEEE International Symposium on Circuits and Systems
  (ISCAS)}, pages 1--5. IEEE, 2018.

\bibitem{Dua:2019}
Dheeru Dua and Casey Graff.
\newblock {UCI} machine learning repository, 2017.

\bibitem{Maaten2008VisualizingDU}
L.~V.~D. Maaten and Geoffrey~E. Hinton.
\newblock Visualizing data using t-sne.
\newblock {\em Journal of Machine Learning Research}, 9:2579--2605, 2008.

\bibitem{Wattenberg2016HowTU}
M.~Wattenberg, F.~Vi{\'e}gas, and I.~Johnson.
\newblock How to use t-sne effectively.
\newblock 2016.

\end{thebibliography}

\end{document}